\documentclass[11pt]{article}
\usepackage{amsfonts}
\usepackage{amsmath,amsthm}

\usepackage{mathrsfs}
\usepackage{cite}
\usepackage{color}
\usepackage[colorlinks,linkcolor=blue,citecolor=blue]{hyperref}
\usepackage{graphicx,subfigure}
\usepackage[numbers,sort&compress]{natbib}
\usepackage[top=1 in,bottom=1 in,left=1.0 in,right=1.0 in]{geometry}

\newcommand{\bblu}{\begin{color}{blue}}
\newcommand{\bred}{\begin{color}{red}}
\newcommand{\ecl}{\end{color}}
\numberwithin{equation}{section}

\newtheorem{proposition}{Proposition}[section]

\def\t#1{\widetilde{#1}}

\def\b#1{\overline{#1}}

\begin{document}

\title{Explicit quasi-periodic solutions to the discrete derivative nonlinear Schr\"odinger equation}
\author{Xiaoxue Xu$^{1}$, ~Wei Liu$^{1}$,  ~Junyi Zhu$^{1}$
\\
{\small ${}^{1}$School of Mathematics and Statistics, Zhengzhou University, Zhengzhou, Henan 450001, P.R. China}\\
{\small E-mail:\ \ jyzhu@zzu.edu.cn}}

\maketitle

\begin{abstract}
Quasi-periodic solutions to the discrete derivative nonlinear Schr\"odinger equation are investigated. A Lax pair of the discrete integrable equation is derived, by which we obtain a nonlinear integrable symplectic map. Resorting to this map, the Riemann theta function expressions for discrete potentials are calculated.
Consequently, we get the explicit solutions for the discrete derivative nonlinear Schr\"odinger equation, whose quasi-periodic properties are illustrated through the cases of genus one and genus two.
\\[2mm]
Keywords: discrete derivative nonlinear Schr\"odinger equation, quasi-periodic solutions, integrable symplectic map, Riemann theta function
\\[1mm]
MSC numbers: 35Q51, 37K60, 39A36
\end{abstract}

%

\section{Introduction }\label{sec-1}

The derivative nonlinear Schr\"odinger (DNLS) equations represent a fundamental class of integrable nonlinear wave equations with profound physical significance
and mathematical richness \cite{Rogister,Johnson,Ichikawa,Kodama,Rangwala,Imai,Yan,Ling}. Among its various forms, the first type DNLS equation is
\begin{equation}\label{KN}
\mathrm{i} u_{t}+u_{xx}-2\mathrm{i}\delta (|u|^2u)_{x}=0,\ \ \delta=\pm1,
\end{equation}
where $\mathrm{i^{2}}=-1, |u|^2=u\bar{u}$, and bar represents complex conjugate. The inverse scattering transform for equation \eqref{KN} was first discussed by Kaup and Newell \cite{Kaup}.
The unreduced DNLS system reads
\begin{subequations}\label{eq:0.0}
\begin{align}
&u_t-u_{xx}+2(u^2v)_x=0,\label{eq:0.0a}\\
&v_t+v_{xx}+2(uv^2)_x=0,\label{eq:0.0b}
\end{align}
\end{subequations}
which yields the DNLS equation \eqref{KN} by replacing $x\to\mathrm{i} x, t\to\mathrm{i} t$ and the reduction $v=-\delta \bar{u}$.

Discrete integrable systems form an important part of the integrable system theory. With the development and progress of discrete integrable systems, extensive research has been devoted to
the discrete counterparts of DNLS system, covering their Lax representations \cite{Tsuchida}, Hamiltonian structures \cite{Tsuchida,Khanizadeh}, B$\mathrm{\ddot{a}}$cklund transformations \cite{Choudhury},
Darboux transformations \cite{Xu,Liu}, and various classes of explicit solutions, including soliton and quasi-periodic solutions \cite{Liu,Geng}. In the present paper, we will investigate the quasi-periodic
solutions for the fully discrete DNLS equation
\begin{subequations}\label{eq:0.1}
 \begin{align}
\Xi_{1} &\equiv \frac{\widetilde{u}}{1+\sqrt{1-4 \widetilde{u}\widetilde{\overline{v}}}}  - \frac{\overline{u}}{1+\sqrt{1-4 \b{u}\b{\t{v}}}} + \left(\frac{\beta_{1}^{2}}{1+\sqrt{1-4u\b{v}}} - \frac{\beta_{2}^{2}}{1+\sqrt{1-4u\t{v}}}\right) u = 0, \label{eq:0.1a}\\
\Xi_{2} &\equiv \frac{\widetilde{v}}{1+\sqrt{1-4u\t{v}}} - \frac{\overline{v} }{1+\sqrt{1-4u\b{v}}}- \left(\frac{\beta_{1}^{2}}{1+\sqrt{1-4 \widetilde{u}\widetilde{\overline{v}}}} - \frac{\beta_{2}^{2}}{1+\sqrt{1-4 \b{u}\b{\t{v}}}}\right) \t{\b{v}} = 0.\label{eq:0.1b}
\end{align}
\end{subequations}
Here, the usual notation, $\widetilde{h}(m_{1},m_{2})=h(m_{1}+1,m_{2}), \overline{h}(m_{1},m_{2})=h(m_{1},m_{2}+1)$, is used, and $\beta_{1},\beta_{2}$ are the lattice parameters. With the help of the Miwa variables
\begin{equation}\label{Miwa}
x=-\frac{m_1}{\beta_1^2}-\frac{m_2}{\beta_2^2},~~
y=-\frac{m_1}{2\beta_1^4}-\frac{m_2}{2\beta_2^4},
\end{equation}
the continuum limit of equation \eqref{eq:0.1} leads to the unreduced DNLS system \eqref{eq:0.0} \cite{XCZ-2022}.

For continuous systems, the algebro-geometric method has proven to be a powerful tool for obtaining quasi-periodic solutions expressed in terms of Riemann theta functions. This approach has been successfully applied to a wide range of hierarchies (see \cite{Matveev} and the references therein). For the discrete settings, however, the development of algebro-geometric constructions is more delicate, requiring the interplay of discrete zero-curvature representations and integrable symplectic maps \cite{XCN,XCZ-2022,XZ,Zhao}. In Ref. \cite{XCZ-2022,XZ}, we have obtained some integrable properties for discrete DNLS equations. However, they have to rely on the investigations of the lattice potential modified Kadomtsev-Petviashvili equation, and the quasi-periodic solutions to equation \eqref{eq:0.1} haven't been calculated. Thus, it is necessary and meaningful to study the explicit quasi-periodic solutions for the fully discrete DNLS equation.

The paper is organized as follows.
In Section \ref{sec-2}, we construct a Lax pair for the discrete DNLS equation \eqref{eq:0.1}, which can be nonlinearized to get an integrable symplectic map.
In Section \ref{sec-3}, according to the theory of Riemann surface, the discrete potentials in terms of Riemann theta functions are derived.
Then in Section \ref{sec-4}, we obtain explicit formulae of quasi-periodic solutions to the discrete DNLS equation \eqref{eq:0.1} by the commutability of integrable symplectic maps.
Some solutions of genus one and genus two cases are depicted. In Section \ref{sec-5}, the conclusions are given.

\section{ The Lax pair and nonlinear integrable symplectic map}\label{sec-2}

In Ref. \cite{XCZ-2022}, the discrete DNLS equation can be derived by using an additional condition, which is automatically satisfied for the lattice potential modified
Kadomtsev--Petviashvili equation. Here, we prove that the discrete DNLS equation \eqref{eq:0.1} has a Lax pair as
\begin{subequations}\label{eq:1.2}
\begin{align}
&\widetilde{\chi}=D^{(\beta_1)}\chi=\left(
    \begin{array}{cc}
           \frac{2 \lambda^2 }{1+\sqrt{1-4u\widetilde{v}}}-\beta_1^2&\frac{2 \lambda u}{1+\sqrt{1-4u\widetilde{v}}}\\
           \frac{2 \lambda \widetilde{v}}{1+\sqrt{1-4u\widetilde{v}}}&1
    \end{array}\right)\chi, \\
&\b{\chi}=D^{(\beta_2)}\chi= \left(\begin{array}{cc}
           \frac{2 \lambda^2}{1+\sqrt{1-4u\b{v}}}-\beta_2^2&\frac{2\lambda  u}{1+\sqrt{1-4u\b{v}}}\\
           \frac{2\lambda \overline{v}}{1+\sqrt{1-4u\b{v}}}&1\end{array}\right)\chi.
\end{align}
\end{subequations}
In fact, we have
\begin{equation*}
\widetilde{D}^{(\beta_2)}D^{(\beta_1)}-\overline{D}^{(\beta_1)}D^{(\beta_2)}
=\left(\begin{array}{cc}
\lambda^4 S_1+\lambda^2 S_2 & \lambda^3 u S_1+2 \lambda \Xi_1 \\
\lambda^3 \overline{\widetilde{v}} S_1 +2 \lambda \Xi_2 & \lambda^2 u \overline{\widetilde{v}} S_1
\end{array}\right),
\end{equation*}
where
\begin{align*}
S_1 =&\frac{16}{4 \beta^2_1 u \overline{\widetilde{v}}- (1+\sqrt{1-4u\b{v}})(1+\sqrt{1-4\t{u}\widetilde{\b{v}}})} \left(\frac{\overline{\t{v}}}{1+\sqrt{1-4 \b{u}\b{\t{v}}}}   \Xi_1  + \frac{u}{1+\sqrt{1-4u\t{v}}}\Xi_2\right) ,\\
S_2 = &\frac{4\widetilde{v}}{1+\sqrt{1-4u\t{v}}}\Xi_1 +\frac{4 \overline{u}}{1+\sqrt{1-4\b{u}\b{\t{v}}}}\Xi_2 - \frac{\beta^2_1}{4} \left( 1+\sqrt{1-4u\t{v}} \right) \left( 1+\sqrt{1-4\b{u}\b{\t{v}} }\right)S_1.
\end{align*}

According to the nonlinearization procedure \cite{Cao1990},  the above Lax pair can be nonlinearized by using the Lax matrix \cite{CY}
\begin{equation}\label{eq:2.10}
L(\lambda)\doteq L(\lambda;p,q)
=\left(\begin{array}{cc}
\frac{1}{2}+Q_\lambda(A^2p,q)&-\lambda Q_\lambda(Ap,p)\\
\lambda Q_\lambda(Aq,q) &-\frac{1}{2}-Q_\lambda(A^2p,q)
\end{array}\right),
\end{equation}
where $Q_{\lambda}(\xi,\eta)=<(\lambda^{2}I-A^{2})\xi,\eta>, A\doteq{\rm diag}(\alpha_1,\cdots,\alpha_N)$ and $<\xi,\eta>=\Sigma_{j=1}^N\xi_j\eta_j$. From the Lax pair \eqref{eq:1.2}, we introduce
\begin{equation}\label{eq:D}
D^{(\beta)}(\lambda)\doteq D^{(\beta)}(\lambda;u,v)
=\left(
    \begin{array}{cc}
           \frac{2 \lambda^2 }{1+\sqrt{1-4u\widehat{v}}}-\beta^2&\frac{2 \lambda u}{1+\sqrt{1-4u\widehat{v}}}\\
           \frac{2 \lambda \widehat{v}}{1+\sqrt{1-4u\widehat{v}}}&1
    \end{array}\right).
\end{equation}
In order to investigate the discrete Lax equation, we calculate the discrete curvature matrix
\begin{equation}\begin{split}\label{eq:2.11}
\widehat{L}(\lambda) D^{(\beta)}(\lambda) -D^{(\beta)}(\lambda) L(\lambda)&=\left(\begin{array}{cc}
\frac{2\beta^{2}\widehat{v}}{1+\sqrt{1-4u\widehat{v}}} & \lambda \\
\frac{4\lambda \widehat{v}^{2}}{(1+\sqrt{1-4u\widehat{v}})^{2}} &\frac{2\widehat{v}}{1+\sqrt{1-4u\widehat{v}}}
\end{array}\right)P_{1}\\
&+\left(\begin{array}{cc}
\frac{2\beta^{2}u}{1+\sqrt{1-4u\widehat{v}}} & 0 \\
\frac{2\lambda}{(1+\sqrt{1-4u\widehat{v}})^{2}} &\frac{2u}{1+\sqrt{1-4u\widehat{v}}}
\end{array}\right)\frac{P_{2}}{\beta},
\end{split}\end{equation}
where
\begin{align}\label{eq:2.12}
P_{1}=\frac{2}{1+\sqrt{1-4u\widehat{v}}}(u+<Ap,p>),\ \ \
P_{2}=\beta(\widehat{v}-<A\widehat{q},\widehat{q}>).
\end{align}

\begin{proposition}\label{P-2-2}
 Under the constraint $u=-<Ap,p>,v=<Aq,q>$, the nonlinearized map $S_{\beta}: \mathbb{R}^{2N} \rightarrow \mathbb{R}^{2N},$
\begin{equation}\label{eq:2.14}
\begin{pmatrix} \widehat{p}_j \\ \widehat{q}_j \end{pmatrix} = \frac{1}{\sqrt{\alpha_{j}^{2}-\beta^{2}}}\left(
    \begin{array}{cc}
           \frac{2 \alpha_{j}^2 }{1+\sqrt{1-4u\widehat{v}}}-\beta^2&\frac{2 \alpha_{j} u}{1+\sqrt{1-4u\widehat{v}}}\\
           \frac{2 \alpha_{j} \widehat{v}}{1+\sqrt{1-4u\widehat{v}}}&1
    \end{array}\right) \begin{pmatrix} p_j \\ q_j \end{pmatrix}, \quad (j=1,\cdots, N),
\end{equation}
is integrable and symplectic, possessing the integrals given by coefficients of the polynomial $F(\lambda; p, q) = \det L(\lambda; p, q)$ in $\lambda$.
\end{proposition}

\begin{proof}
Under this constraint, the invariance $\widehat{F}(\lambda; p, q) = F(\lambda; p, q)$ is got by equation \eqref{eq:2.11}. Moreover, the symplectic formula holds,
\begin{align}\label{eq:2.5}
\sum_{j=1}^{N} \left( \mathrm{d}\widehat{p}_{j} \wedge \mathrm{d}\widehat{q}_{j} - \mathrm{d} p_{j} \wedge \mathrm{d} q_{j} \right) = \frac{1}{2}  \mathrm{d} \Omega=0,
\end{align}
with
\begin{equation}\begin{split}\label{eq:2.13}
\Omega=&-\frac{2P_{2}}{\beta}\mathrm{d}\Big(\frac{u}{1+\sqrt{1-4u\widehat{v}}}\Big)-(1+\sqrt{1-4u\widehat{v}})\mathrm{d}\Big(\frac{1}{1+\sqrt{1-4u\widehat{v}}}\Big)\\
&+\frac{1+\sqrt{1-4u\widehat{v}}}{2}P_{1}\Big(\frac{-8\widehat{v}^{2}}{(1+\sqrt{1-4u\widehat{v}})^{2}}\mathrm{d}\big(\frac{u}{1+\sqrt{1-4u\widehat{v}}}\big)+2\mathrm{d}\big(\frac{\widehat{v}}{1+\sqrt{1-4u\widehat{v}}}\big)\Big).
\end{split}\end{equation}
\end{proof}

\section{Riemann theta function expressions of discrete potentials }\label{sec-3}

From the nonlinear integrable symplectic map $S_{\beta}$, we define a discrete phase flow by iteration
\begin{equation}\label{3.1}
(p(m),\,q(m))=(S_{\beta})^m(p(0),q(0))=S_{\beta}^m(p(0),q(0)),
\end{equation}
where $(p(0),q(0))$ is any given point in $\mathbb{R}^{2N}$. On the $S^{m}_{\beta}$-flow, we reformulate the discrete potentials as
\begin{equation}\label{3.2}
u(m)=-<Ap(m),p(m)>,\ \ v(m)=<Aq(m),q(m)>.
\end{equation}
By equation \eqref{eq:2.11} we have the commutative relation
\begin{equation}\label{eq:3.3}
L_{m+1}(\lambda)D^{(\beta)}_m(\lambda)=D^{(\beta)}_m(\lambda)L_m(\lambda),
\end{equation}
with $L_m(\lambda)=L(\lambda; p(m),\,q(m)), D^{(\beta)}_m(\lambda)=D^{(\beta)}(\lambda;u(m), v(m))$.

\noindent Consider the linear operator $L_m(\lambda)$ with two distinct eigenvalues $\pm \rho_{\lambda}=\pm\sqrt{-F(\lambda)}$.
The corresponding eigenvectors are denoted by $h_{\pm}(m,\lambda)=(h_{\pm}^{(1)}(m,\lambda), h_{\pm}^{(2)}(m,\lambda))^T$. When $m=0$ we choose the normalized eigenvector
\begin{subequations}
\begin{equation}\label{eq:3.4a}
h_{\pm}(0,\lambda)=(c_\lambda^\pm,1)^T,
\end{equation}
where
\begin{equation}\label{eq:3.4b}
c_\lambda^\pm=\frac{L_0^{11}(\lambda)\pm
\rho_{\lambda}}{L_0^{21}(\lambda)}
=\frac{-L_0^{12}(\lambda)}{L_0^{11}(\lambda)\mp \rho_{\lambda}}.
\end{equation}
\end{subequations}
Let
\begin{equation}\label{eq:3.5}
M(m,\lambda)=D^{(\beta)}_{m-1}(\lambda)D^{(\beta)}_{m-2}(\lambda)\cdots D^{(\beta)}_0(\lambda),\ \ M(0,\lambda)=I,
\end{equation}
where $I$ is the usual unit matrix. Then we obtain
\begin{subequations}
\begin{equation}\label{eq:3.6a}
h_{\pm}(m,\lambda)=M(m,\lambda){c_\lambda^\pm\choose 1},
\end{equation}
\begin{equation}\label{eq:3.6b}
h_{\pm}(m+1,\lambda)=D_m(\lambda)h_{\pm}(m,\lambda).
\end{equation}
\end{subequations}

Our aim is to investigate these eigenvectors $h_{\pm}(m,\lambda)$ under the theory of Riemann surface \cite{Farkas,Griffiths,Mumford}. Now we consider a hyperelliptic curve
\begin{equation}\label{R-curve}
\mathcal R:~~ \xi^2=R(\zeta)=\prod_{j=1}^N(\zeta-\alpha_j^2)(\zeta-\lambda_j^2),\ \ \zeta=\lambda^2,
\end{equation}
where $\lambda_j^2$ are the zeros of $F(\lambda)$. The genus of $\mathcal R$ is $g=N-1$. For values of $\zeta$ not corresponding to a branch point, there are two points $\mathfrak p_+(\zeta),\mathfrak p_-(\zeta)$ on $\mathcal R$,
\begin{equation}\label{eq:3.8}
\mathfrak p_+(\zeta)=\big(\zeta,\,\xi=\sqrt{R(\zeta)}\big),\quad ~
\mathfrak p_-(\zeta)=\big(\zeta,\,\xi=-\sqrt{R(\zeta)}\big),
\end{equation}
and in particular, for the infinity we denote the two corresponding points on $\mathcal R$ by $\infty_+,\,\infty_-$. Thus we obtain four meromorphic functions on $\mathcal{R}$,
\begin{equation}\label{eq:3.9}
\begin{split}
&\mathfrak{h}^{(1)}\big(m,\mathfrak p_+(\lambda^2)\big)=\lambda h_+^{(1)}(m,\lambda),\quad
\mathfrak{h}^{(1)}\big(m, \mathfrak p_-(\lambda^2)\big)=\lambda h_-^{(1)}(m,\lambda),\\
&\mathfrak{h}^{(2)}\big(m,\mathfrak p_+(\lambda^2)\big)=h_+^{(2)}(m,\lambda),\quad
\mathfrak{h}^{(2)}\big(m, \mathfrak p_-(\lambda^2)\big)=h_-^{(2)}(m,\lambda).
\end{split}\end{equation}

\begin{proposition}\label{P-3-1}
The divisors of $\mathfrak h^{(1)}(m,\mathfrak p)$,\,$\mathfrak h^{(2)}(m,\mathfrak p)$ on $\mathcal{R}$ are written as
\begin{subequations}\label{eq:3.10}
\begin{align}
&\mathcal{D}(\mathfrak h^{(1)}(m,\mathfrak p))=\sum_{j=1}^g\Big(\mathfrak p\big(\mu_j^2(m)\big)-\mathfrak
p\big(\nu_j^2(0)\big)\Big)+\{\mathfrak o_-\}+m\{\mathfrak p(\beta^2)\}-(m+1)\{\infty_+\},\label{eq:3.10a}\\
&\mathcal{D}(\mathfrak h^{(2)}(m,\mathfrak p))=\sum_{j=1}^g\Big(\mathfrak p\big(\nu_j^2(m)\big)-\mathfrak
p\big(\nu_j^2(0)\big)\Big)+m\{\mathfrak p(\beta^2)\}-m\{\infty_+\},\label{eq:3.10b}
\end{align}\end{subequations}
where $\mathfrak o_-=\big(\zeta=0,\,\xi=-\sqrt{R(0)}\big)$, and $\mu_j^{2},\,\nu_j^{2}$ are roots of $L^{12}(\lambda)$ and $L^{21}(\lambda)$ respectively.
\end{proposition}

\begin{proof}
From \eqref{eq:3.3}, \eqref{eq:3.4b}, \eqref{eq:3.5}, \eqref{eq:3.6a} and \eqref{eq:3.9} we have
\begin{subequations}\label{eq:3.11}
\begin{align}
&\mathfrak{h}^{(1)}\big(m,\mathfrak p_+(\lambda^2)\big)\cdot
\mathfrak{h}^{(1)}\big(m, \mathfrak p_-(\lambda^2)\big)
=-\zeta(\zeta-\beta^2)^m\frac{u(m)}{v(0)}\prod_{j=1}^{g}\frac{\zeta-\mu_j^2(m)}{\zeta-\nu_j^2(0)},
\label{eq:3.11a}\\
&\mathfrak{h}^{(2)}\big(m,\mathfrak p_+(\lambda^2)\big)\cdot
\mathfrak{h}^{(2)}\big(m, \mathfrak p_-(\lambda^2)\big)
=(\zeta-\beta^2)^m\frac{v(m)}{v(0)}\prod_{j=1}^{g}\frac{\zeta-\nu_j^2(m)}{\zeta-\nu_j^2(0)}.\label{eq:3.11b}
\end{align}
\end{subequations}
As $\zeta\sim\infty$, we obtain the following asymptotic behaviors:
\begin{subequations}\label{eq:3.12}
\begin{align}
&\mathfrak{h}^{(1)}\big(m,\mathfrak p_+(\lambda^2)\big)
=\frac{2^m}{v(0) \big(1+\sqrt{1-4u(0)v(1)}\big) \cdots \big(1+\sqrt{1-4u(m-1) v(m)}\big)} \zeta^{m+1}\big(1+O(\zeta^{-1})\big),\label{eq:3.12a}\\
&\mathfrak{h}^{(1)}\big(m, \mathfrak p_-(\lambda^2)\big)
=-\frac{u(m)}{2^m} {\big(1+\sqrt{1-4u(0) v(1)}\big) \cdots \big(1+\sqrt{1-4u(m-1) v(m)}\big)} \big(1+O(\zeta^{-1})\big),\label{eq:3.12b}\\
&\mathfrak{h}^{(2)}\big(m,\mathfrak p_+(\lambda^2)\big)
=\frac{2^m v(m)}{v(0) \big(1+\sqrt{1-4u(0) v(1)}\big) \cdots \big(1+\sqrt{1-4u(m-1) v(m)}\big)} \zeta^{m}\big(1+O(\zeta^{-1})\big),\label{eq:3.12c}\\
&\mathfrak{h}^{(2)}\big(m, \mathfrak p_-(\lambda^2)\big)
=\frac{1}{2^m} {\big(1+\sqrt{1-4u(0) v(1)}\big) \cdots \big(1+\sqrt{1-4u(m-1) v(m)}\big)} \big(1+O(\zeta^{-1})\big).\label{eq:3.12d}
\end{align}\end{subequations}
by using \eqref{eq:3.4b}, \eqref{eq:3.5}, \eqref{eq:3.6a}. Similarly, we get ($\zeta\sim 0$)
\begin{align}\label{eq:3.13}
\mathfrak{h}^{(1)}\big(m, \mathfrak p_-(\lambda^2)\big)
=\zeta\frac{<A^{-1}q(0),q(0)>
u(m)}{(1-2<p(0),q(0)>)v(0)}\left(\prod_{j=1}^{g}\frac{\mu_j^2(m)}{\nu_j^2(0)}\right)\big(1+O(\zeta)\big).
\end{align}
Thus \eqref{eq:3.10} is obtained by the definition of divisor.
\end{proof}
\noindent Let us introduce the Abel-Jacobi variables
\begin{equation}\label{eq:3.14}
\vec{\psi}(m)=\mathcal{A}\Big({\sum}_{j=1}^g\mathfrak p\big(\mu_j^2(m)\big)\Big), \quad
\vec{\phi}(m)=\mathcal{A}\Big({\sum}_{j=1}^g\mathfrak p\big(\nu_j^2(m)\big)\Big).
\end{equation}
by the Abel map $\mathcal{A}$ \cite{Farkas,Griffiths,Mumford}. In the Jacobi variety we have
\begin{subequations}\label{eq:3.15}
\begin{align}
\label{eq:3.15a}
&\vec\psi(m)\equiv\vec\phi(0)+m\vec\Omega_\beta+\vec\Omega_0,\quad({\rm mod}\,\mathcal T),\\
\label{eq:3.15b}
&\vec\phi(m)\equiv\vec\phi(0)+m\vec\Omega_\beta,\quad({\rm mod}\,\mathcal T),\\
\label{eq:3.15c}
&\vec\Omega_\beta=\int_{\mathfrak p(\beta^2)}^{\infty_+}\vec\omega,\quad
\vec\Omega_0=\int_{\mathfrak o_-}^{\infty_+}\vec\omega,
\end{align}\end{subequations}
where $\mathcal T$ is a period lattice and $\vec\omega$ is the normalized basis of holomorphic differentials of $\mathcal R$. Consequently we obtain the Riemann theta function expressions \cite{Farkas,Griffiths,Mumford}:
\begin{subequations}\label{eq:3.16}
\begin{align}\label{eq:3.16a}
&\mathfrak h^{(1)}(m,\mathfrak p)=C_m^{(1)}
\frac{\theta(-\mathcal A(\mathfrak p)+\vec\psi(m)+\vec K;B)}
{\theta(-\mathcal A(\mathfrak p)+\vec\phi(0)+\vec K;B)}
\exp\int_{\mathfrak p_0}^{\mathfrak p}(m\,\omega[\mathfrak p(\beta^2),\infty_+]
+\omega[\mathfrak o_-,\infty_+]),\\
\label{eq:3.16b}
&\mathfrak h^{(2)}(m,\mathfrak p)=C_m^{(2)}
\frac{\theta(-\mathcal A(\mathfrak p) +\vec\phi(m)+\vec K;B)}
{\theta(-\mathcal A(\mathfrak p)+\vec\phi(0)+\vec K;B)}
\exp\int_{\mathfrak p_0}^{\mathfrak p}m\,\omega[\mathfrak p(\beta^2),\infty_+],
\end{align}
\end{subequations}
where $C_m^{(1)}$ and $C_m^{(2)}$ are constant factors, $\omega$ is the dipole and $B$ is a symmetric matrix with positive definite imaginary part. Moreover the elements of the Riemann constant vector $\vec K=(\vec K_{1},\ldots,\vec K_{g})^{T}$ is
\begin{align}
\vec K_{j}=\displaystyle\frac{1+B_{jj}}{2}-\sum^{g}_{\begin{subarray}{c}k=1\\ k\neq j \end{subarray}}\int_{a_{k}}\mathcal A_{j}(\mathfrak p)\ \omega_{k},
\ \ j=1,\ldots,g,
\end{align}
with $\mathcal{A}_{j}$ being the $j$-th component of the Abel map $\mathcal A$.
\begin{proposition}\label{P-3-2}
The Riemann theta function expressions for discrete potentials $u(m)$ and $v(m)$ are
\begin{subequations}\label{eq:3.20}
\begin{equation}\label{eq:3.20a}
\begin{split}
    u(m) &= u(0)  \frac{\theta^2(\vec\psi(0)+\vec K + \vec\eta_{\mathfrak o_+};B) \cdot \theta(\vec\psi(m)+\vec K+\vec\eta_{\infty_+};B)\cdot \theta(\vec\psi(m)+\vec K+\vec\eta_{\infty_-};B)}
{\theta^2(\vec\psi(m)+\vec K + \vec\eta_{\mathfrak o_+};B) \cdot \theta(\vec\psi(0)+\vec K+\vec\eta_{\infty_+};B)\cdot \theta(\vec\psi(0)+\vec K+\vec\eta_{\infty_-};B)}\\& \quad
  \cdot (\beta^{4}\gamma_{\beta}^+)^m \cdot \exp \big(\int_{\mathfrak o_+}^{\infty_-} m\omega[\mathfrak p(\beta^2),\infty_+] + \int_{\mathfrak o_+}^{\mathfrak{p}_0} m \omega[\mathfrak p(\beta^2),\infty_+]\big),
\end{split}
\end{equation}
\begin{equation}\label{eq:3.20b}
\begin{split}
v(m) &=  v(0)\frac{\theta^2(\vec\phi(0)+\vec K + \vec\eta_{\mathfrak o_-};B) \cdot \theta(\vec\phi(m)+\vec K+\vec\eta_{\infty_-};B)\cdot \theta(\vec\phi(m)+\vec K+\vec\eta_{\infty_+};B)}
{\theta^2(\vec\phi(m)+\vec K + \vec\eta_{\mathfrak o_-};B) \cdot \theta(\vec\phi(0)+\vec K+\vec\eta_{\infty_-};B)\cdot \theta(\vec\phi(0)+\vec K+\vec\eta_{\infty_+};B)}\\
& \quad \cdot (\gamma_{\beta}^+)^{m} \cdot \exp \big(\int_{\mathfrak o_-}^{\infty_-} m\omega[\mathfrak p(\beta^2),\infty_+] +\int_{\mathfrak o_-}^{\mathfrak{p}_0} m\omega[\mathfrak p(\beta^2),\infty_+]\big) ,
\end{split}
\end{equation}
\end{subequations}
where $\vec\eta_{\infty_\pm}=-\mathcal A(\infty_\pm), \vec \eta_{\mathfrak {o_\pm}}=-\mathcal A(\mathfrak {o_\pm})$ and
\begin{align}\label{eq:3.21A}
&\gamma_{\beta}^{+}=\lim_{\mathfrak p\to \infty_+}\frac{1}{\zeta}
\exp\int_{\mathfrak p_0}^{\mathfrak p} \omega[\mathfrak p(\beta^2),\infty_+],\\
\label{eq:3.21B}
&\omega[\mathfrak p(\beta^2),\infty_+]=
\left(\zeta^{g}+\frac{\xi+\sqrt{R(\beta^{2})}}{\zeta-\beta^{2}}\right)
\frac{\mathrm{d}\zeta}{2\sqrt{R(\zeta)}}.
\end{align}
\end{proposition}

\begin{proof}
By using equation \eqref{eq:3.6b} we get
\begin{equation}\label{eq:3.17}
\mathfrak h^{(1)}(m+1,\mathfrak o_+)=- \beta^2 \mathfrak h^{(1)}(m,\mathfrak o_+),
\end{equation}
at the point $\mathfrak o_+ = (\zeta=0, \xi= \sqrt{R(0)})$. Substituting \eqref{eq:3.16a} into \eqref{eq:3.17}, we obtain
\begin{equation}\label{eq:3.18}
\frac{C_{m+1}^{(1)}}{C_m^{(1)}}=-\beta^2
\frac{\theta(\vec\psi(m)+\vec K + \vec\eta_{\mathfrak o_+};B)}
{\theta(\vec\psi(m+1)+\vec K + \vec\eta_{\mathfrak o_+};B)} \cdot
\exp\int_{\mathfrak{o}_{+}}^{\mathfrak p_0}\omega[\mathfrak p(\beta^2),\infty_+].
\end{equation}
Similarly, by \eqref{eq:3.6b} and \eqref{eq:3.16b} we arrive at
\begin{equation}\label{eq:3.19}
\frac{C_{m+1}^{(2)}}{C_m^{(2)}}=
\frac
{\theta(\vec\phi(m)+\vec K+\vec\eta_{\mathfrak o_-};B)}{\theta(\vec\phi(m+1)+\vec K +\vec\eta_{\mathfrak o_-};B)} \cdot
\exp\int_{\mathfrak o_-}^{\mathfrak p_0}\omega[\mathfrak p(\beta^2),\infty_+].
\end{equation}
Hence from \eqref{eq:3.12} and \eqref{eq:3.16} we deduce \eqref{eq:3.20}.
\end{proof}

\section{Quasi-periodic solutions to the discrete DNLS equation }\label{sec-4}

Applying the theory in Section \ref{sec-2} to the parameters $\beta=\beta_1,\,\beta_2$ respectively, we get two integrable maps $S_{\beta_1},\,S_{\beta_2}$ which possess the same Liouville
set of integrals. Thus, the corresponding phase flows $S_{\beta_{1}}^{m_{1}}$ and $S_{\beta_{2}}^{m_{2}}$ commute \cite{Arnold,Bruschi,Veselov}. This guarantees $\widetilde{D}^{(\beta_2)}D^{(\beta_1)}=\overline{D}^{(\beta_1)}D^{(\beta_2)}$. Besides, the function $\big(p(m_{1},m_{2}),q(m_{1},m_{2})\big)=S^{m_{1}}_{\beta_{1}}S_{\beta_{2}}^{m_{2}}(p_{0},q_{0})$ is well-defined and the evolution of Abel-Jacobi variables along the above two discrete flows gives
\begin{equation}\label{eq:0.01}
\begin{split}
&\vec{\phi}(m_{1},m_{2})= \vec{\phi}(0,0)+m_{1}\vec{\Omega}_{\beta_{1}} +m_{2}\vec{\Omega}_{\beta_{2}},\\
&\vec{\psi}(m_{1},m_{2})= \vec{\phi}(0,0)+m_{1}\vec{\Omega}_{\beta_{1}} +m_{2}\vec{\Omega}_{\beta_{2}}+\vec{\Omega}_{0}.
\end{split}
\end{equation}
As a consequence, the discrete DNLS equation \eqref{eq:0.1} admits the quasi-periodic solutions in terms of Riemann theta functions:

\begin{subequations}\label{eq:4.1}
\begin{align}
u(m_1,m_2)=&u(0,0)\frac{\theta(\vec\phi(0,0)+m_1\vec\Omega_{\beta_1}+m_2\vec\Omega_{\beta_2}+\vec\Omega_0
+\vec K+\vec\eta_{\infty_+};B)\cdot\theta^2(\vec\phi(0,0)+\vec\Omega_0
+\vec K+\vec\eta_{\mathfrak{o}_+};B)}
{\theta(\vec\phi(0,0)+\vec\Omega_0
+\vec K+\vec\eta_{\infty_+};B)\cdot\theta^2(\vec\phi(0,0)+m_1\vec\Omega_{\beta_1}+m_2\vec\Omega_{\beta_2}+\vec\Omega_0
    +\vec K+\vec\eta_{\mathfrak{o}_+};B)}\nonumber
\\
&\cdot\frac{\theta(\vec\phi(0,0)+m_1\vec\Omega_{\beta_1}+m_2\vec\Omega_{\beta_2}+\vec\Omega_0
+\vec K+\vec\eta_{\infty_-};B)}{\theta(\vec\phi(0,0)+\vec\Omega_0
+\vec K+\vec\eta_{\infty_-};B)}\cdot (\beta_1^{4}\gamma_{\beta_1}^+)^{m_1} \cdot (\beta_2^{4}\gamma_{\beta_2}^+)^{m_2}  \nonumber
\\
&\cdot \exp\Big[ \int_{\mathfrak o_+}^{\mathfrak{p}_0} ( m_1\omega[\mathfrak p(\beta_1^2),\infty_+]+m_2\omega[\mathfrak p(\beta_2^2),\infty_+])+\int_{\mathfrak o_+}^{\infty_-}( m_1 \omega[\mathfrak p(\beta_1^2),\infty_+]+m_2 \omega[\mathfrak p(\beta_2^2),\infty_+])\Big] ,\label{eq:4.1a}
\\
v (m_1,m_2 )=& v(0,0)\frac{\theta(\vec\phi(0,0)+m_1\vec\Omega_{\beta_1}+m_2\vec\Omega_{\beta_2}+\vec K+\vec\eta_{\infty_-};B)\cdot\theta^2(\vec\phi(0,0)
+\vec K+\vec\eta_{\mathfrak{o}_-};B)}
{\theta(\vec\phi(0,0)
+\vec K+\vec\eta_{\infty_-};B)\cdot\theta^2(\vec\phi(0,0)+m_1\vec\Omega_{\beta_1}+m_2\vec\Omega_{\beta_2}
    +\vec K+\vec\eta_{\mathfrak{o}_-};B)}\nonumber
\\
& \cdot\frac{\theta(\vec\phi(0,0)+m_1\vec\Omega_{\beta_1}+m_2\vec\Omega_{\beta_2}
+\vec K+\vec\eta_{\infty_+};B)}{\theta(\vec\phi(0,0)
+\vec K+\vec\eta_{\infty_+};B)}\cdot(\gamma_{\beta_1}^+)^{m_1} \cdot (\gamma_{\beta_2}^+)^{m_2}  \nonumber
\\
& \cdot \exp\Big[ \int_{\mathfrak o_-}^{\mathfrak{p}_0} ( m_1\omega[\mathfrak p(\beta_1^2),\infty_+]+m_2\omega[\mathfrak p(\beta_2^2),\infty_+])+\int_{\mathfrak o_-}^{\infty_-}( m_1\omega[\mathfrak p(\beta_1^2),\infty_+]+m_2\omega[\mathfrak p(\beta_2^2),\infty_+])\Big],\label{eq:4.1b}
\end{align}
\end{subequations}
where $\gamma_{\beta_j}^+,\omega[\mathfrak p(\beta_j^2),\infty_+],j=1,2$ are got from \eqref{eq:3.21A} and \eqref{eq:3.21B} by taking $\beta=\beta_1,\beta_2$ respectively.

In order to show the quasi-periodic oscillatory behavior of the solution, we will plot only the Riemann theta-function contribution by omitting the exponential factors in \eqref{eq:4.1}. This is justified because, after taking the logarithm, these factors become linear functions of the discrete variables $m_1,m_2$, which only add a global linear trend and do not affect the oscillatory frequencies or the quasi-periodic structure. The remaining theta quotient part captures the essential nonlinear oscillation modes. In other words, removing the exponential factors can be considered as a form of normalization, and the background is normalized for the purpose of clearer visualization. This normalization does not alter the quasi-periodic nature of the solution, as the frequencies are determined entirely by the periods of the theta functions.


In the case of $g = 1$, the Riemann theta function $\theta$ reduces to the Jacobi theta function $\vartheta_{3}$.
Here we investigate
\begin{subequations}\label{eq:4.2}
\begin{align}
u^{(1)}(m_1,m_2)=&\frac{\vartheta_{3}(\phi(0,0)+m_1\Omega_{\beta_1}+m_2\Omega_{\beta_2}+\Omega_0
+ K+\eta_{\infty_-}|B_{11})}{\vartheta_{3}(\phi(0,0)+\Omega_0
+K+\eta_{\infty_-}|B_{11})}\nonumber
\\
&\frac{\vartheta_{3}(\phi(0,0)+m_1\Omega_{\beta_1}+m_2\Omega_{\beta_2}+\Omega_0
+K+\eta_{\infty_+}|B_{11})\cdot\vartheta_{3}^2(\phi(0,0)+\Omega_0
+K+\eta_{\mathfrak{o}_+}|B_{11})}
{\vartheta_{3}(\phi(0,0)+\Omega_0
+K+\eta_{\infty_+}|B_{11})\cdot\vartheta_{3}^2(\phi(0,0)+m_1\Omega_{\beta_1}+m_2\Omega_{\beta_2}+\Omega_0
    +K+\eta_{\mathfrak{o}_+}|B_{11})},\label{eq:4.2a}
\\
v^{(1)}(m_1,m_2 )=&\frac{\vartheta_{3}(\phi(0,0)+m_1\Omega_{\beta_1}+m_2\Omega_{\beta_2}
+K+\eta_{\infty_+}|B_{11})}{\vartheta_{3}(\phi(0,0)
+K+\eta_{\infty_+}|B_{11})}\nonumber
\\
&\frac{\vartheta_{3}(\phi(0,0)+m_1\Omega_{\beta_1}+m_2\Omega_{\beta_2}+K+\eta_{\infty_-}|B_{11})\cdot\vartheta_{3}^2(\phi(0,0)
+K+\eta_{\mathfrak{o}_-}|B_{11})}
{\vartheta_{3}(\phi(0,0)
+K+\eta_{\infty_-}|B_{11})\cdot\vartheta_{3}^2(\phi(0,0)+m_1\Omega_{\beta_1}+m_2\Omega_{\beta_2}
    +K+\eta_{\mathfrak{o}_-}|B_{11})}.\label{eq:4.2b}
\end{align}
\end{subequations}
To illustrate the solutions, we set the hyperelliptic curve
\begin{equation}\label{eq:C1}
\xi^2=R_{1}(\zeta)=(\zeta-1)(\zeta-3)(\zeta-5)(\zeta-8).
\end{equation}
By using numerical integrals calculated by Maple, we get the normalized holomorphic differential
\begin{equation}\label{eq:w1}
\omega=\displaystyle\frac{1.304674569\,\mathrm{i}\mathrm{d}\zeta}{2\sqrt{(\zeta-1)(\zeta-3)(\zeta-5)(\zeta-8)}},
\end{equation}
and
\begin{equation}\label{eq:bk1}
B_{11}=1.210908403\,\mathrm{i},\ \ K= 0.500000000+0.605454201\,\mathrm{i}.
\end{equation}
By taking $\beta_1^2=10,~\beta_2^2=13$ and $\mathfrak{p}_0=(-3.0, \sqrt{R_{1}(-3.0)})$, we have
\begin{equation}\label{eq:integrals1}
\begin{split}
&\eta_{\mathfrak o_+}=-\int_{\mathfrak{p}_{0}}^{\mathfrak o_+}\omega=-0.084846526\,\mathrm{i},\quad \eta_{\mathfrak o_-}=-\int_{\mathfrak{p}_{0}}^{\mathfrak o_-}\omega=-0.391543538\,\mathrm{i}\\
&\eta_{\infty_+}=-\int_{\mathfrak{p}_{0}}^{\infty_+}\omega=0.094041449\,\mathrm{i},\quad \eta_{\infty_-}=-\int_{\mathfrak{p}_{0}}^{\infty_-}\omega=-0.194149314-0.465152297\,\mathrm{i}\\
&\Omega_{\beta_{1}}=\int_{\mathfrak p(\beta_{1}^2)}^{\infty_+}\omega=0.123456359\,\mathrm{i},\quad \Omega_{\beta_{2}}=\int_{\mathfrak p(\beta_{2}^2)}^{\infty_+}\omega=0.076991339\,\mathrm{i},\ \
\Omega_0=\int_{\mathfrak o_-}^{\infty_+}\omega=-0.485619804\,\mathrm{i}.
\end{split}
\end{equation}
In the genus one case, we set $\phi(0,0)=0$. Then solutions for the discrete DNLS equation \eqref{eq:0.1} is displayed in Figure \ref{Fig-1}. Although we plotted part of the Riemann theta function, Figure \ref{Fig-1} still exhibits a background trend. The reason is that its period average may not vanish.

\begin{figure}[htbp]
  \centering
  \includegraphics[width=7cm]{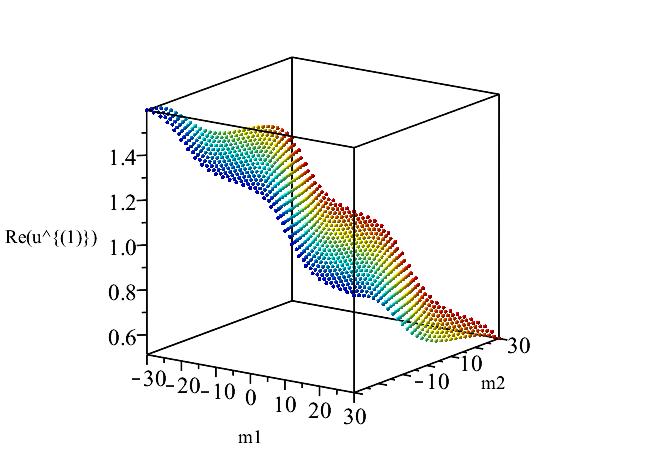}\quad
  \includegraphics[width=7cm]{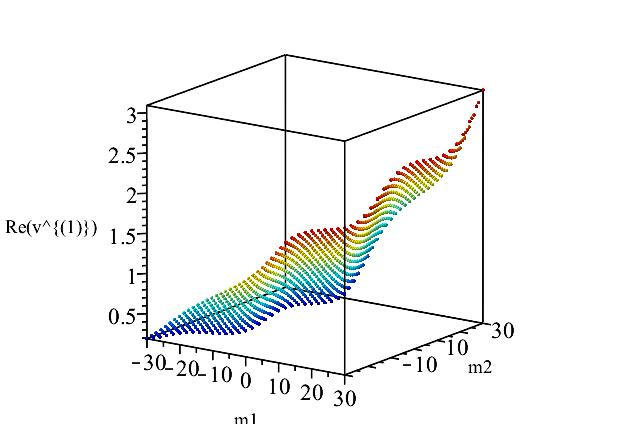}\\
  \caption{For $g=1$, the discrete grid of real part of $u^{(1)}$ and $v^{(1)}$ in \eqref{eq:4.2}}\label{Fig-1}
\end{figure}

In the case of $g = 2$, we consider
\begin{subequations}\label{eq:4.3}
\begin{align}
u^{(2)}(m_1,m_2)=&\frac{\theta(\vec\phi(0,0)+m_1\vec\Omega_{\beta_1}+m_2\vec\Omega_{\beta_2}+\vec\Omega_0
+\vec K+\vec\eta_{\infty_-};B)}{\theta(\vec\phi(0,0)+\vec\Omega_0
+\vec K+\vec\eta_{\infty_-};B)}\nonumber
\\
&\cdot\frac{\theta(\vec\phi(0,0)+m_1\vec\Omega_{\beta_1}+m_2\vec\Omega_{\beta_2}+\vec\Omega_0
+\vec K+\vec\eta_{\infty_+};B)\cdot\theta^2(\vec\phi(0,0)+\vec\Omega_0
+\vec K+\vec\eta_{\mathfrak{o}_+};B)}
{\theta(\vec\phi(0,0)+\vec\Omega_0
+\vec K+\vec\eta_{\infty_+};B)\cdot\theta^2(\vec\phi(0,0)+m_1\vec\Omega_{\beta_1}+m_2\vec\Omega_{\beta_2}+\vec\Omega_0
    +\vec K+\vec\eta_{\mathfrak{o}_+};B)}
,\label{eq:4.3a}
\\
v^{(2)}(m_1,m_2 )=&\frac{\theta(\vec\phi(0,0)+m_1\vec\Omega_{\beta_1}+m_2\vec\Omega_{\beta_2}
+\vec K+\vec\eta_{\infty_+};B)}{\theta(\vec\phi(0,0)
+\vec K+\vec\eta_{\infty_+};B)}\nonumber
\\
&\cdot\frac{\theta(\vec\phi(0,0)+m_1\vec\Omega_{\beta_1}+m_2\vec\Omega_{\beta_2}+\vec K+\vec\eta_{\infty_-};B)\cdot\theta^2(\vec\phi(0,0)
+\vec K+\vec\eta_{\mathfrak{o}_-};B)}
{\theta(\vec\phi(0,0)
+\vec K+\vec\eta_{\infty_-};B)\cdot\theta^2(\vec\phi(0,0)+m_1\vec\Omega_{\beta_1}+m_2\vec\Omega_{\beta_2}
    +\vec K+\vec\eta_{\mathfrak{o}_-};B)}.
\label{eq:4.1b}
\end{align}
\end{subequations}
Let the associated hyperelliptic curve has the following form
\begin{equation}\label{eq:C2}
\xi^2=R_{2}(\zeta)=(\zeta-1)(\zeta-3)(\zeta-5)(\zeta-8)(\zeta-10)(\zeta-12).
\end{equation}
The normalized basis of holomorphic differentials as well as the periodic matrix are chosen as
\begin{equation}\label{eq:w2}
\vec{\omega}=\left(\begin{array}{cc}
-2.647369147\,\mathrm{i}&28.636475437\,\mathrm{i}\\
2.647369147\,\mathrm{i}&-5.779323484\,\mathrm{i}
\end{array}\right)\left(\begin{array}{c}
        \zeta\mathrm{d}\zeta/2\sqrt{R_{2}(\zeta)} \\
         \mathrm{d}\zeta/2\sqrt{R_{2}(\zeta)}
    \end{array}\right),
\end{equation}

\begin{equation}\label{eq:B2}
B=\left(\begin{array}{cc}
1.208908908\,\mathrm{i}&0.333917643\,\mathrm{i}\\
0.333917643\,\mathrm{i}&1.208908908\,\mathrm{i}
\end{array}\right).
\end{equation}
In addition, we take
\begin{equation}
\begin{split}
&\vec{\phi}(0,0)=(0.05,\;0.05)^{\mathrm T},\quad
\vec{\psi}(0,0)=\left(0.05000000-0.4869800\,\mathrm{i},\;0.05000000+0.1265994\,\mathrm{i}\right)^{\mathrm T}.\\
&\vec{\Omega}_{\beta_1}=
\left(
-0.08165449\,\mathrm{i},\;
0.1797596\,\mathrm{i}
\right)^{\mathrm T},
~~
\vec{\Omega}_{\beta_2}
=
\left(
-0.06418381\,\mathrm{i},\;
0.1116356\,\mathrm{i}
\right)^{\mathrm T},
\\
&\vec{\Omega}_{0}
=
\left(
-0.4869800\,\mathrm{i},\;
0.1265994\,\mathrm{i}
\right)^{\mathrm T},
~~
\vec{K}
=
\left(
0.5000000-0.07165829\,\mathrm{i},\;
0.5000000+0.6502627\,\mathrm{i}
\right)^{\mathrm T},
\\
&\vec{\eta}_{\mathfrak o_+}
=
\left(
-0.08507253\,\mathrm{i},\;
0.02340195\,\mathrm{i}
\right)^{\mathrm T},
~~
\vec{\eta}_{\mathfrak o_-}
=
\left(
-0.3921587\,\mathrm{i},\;
0.06821486\,\mathrm{i}
\right)^{\mathrm T},
\\
&\vec{\eta}_{\infty_+}
=
\left(
-0.09468568\,\mathrm{i},\;
0.05825246\,\mathrm{i}
\right)^{\mathrm T},
~~
\vec{\eta}_{\infty_-}
=
\left(
-0.5719160\,\mathrm{i},\;
0.1498692\,\mathrm{i}
\right)^{\mathrm T}.
\end{split}
\end{equation}
The local two-dimensional structure of the unreduced theta quotient is displayed in Figure \ref{Fig-4}.
The relatively small computational window is chosen to prevent strong local amplification associated with small theta denominators from obscuring the underlying local structure. Their geometric differences mainly originate from the distinct theta-function combinations entering the two components. Both exhibit coherent two-dimensional modulation, but with different local curvature and amplitude distributions.
\begin{figure}[htbp]
  \centering
  \includegraphics[width=7cm]{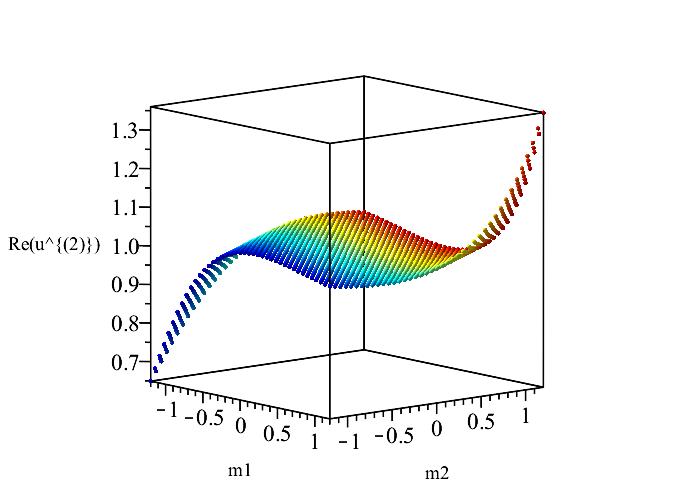}\quad
  \includegraphics[width=7cm]{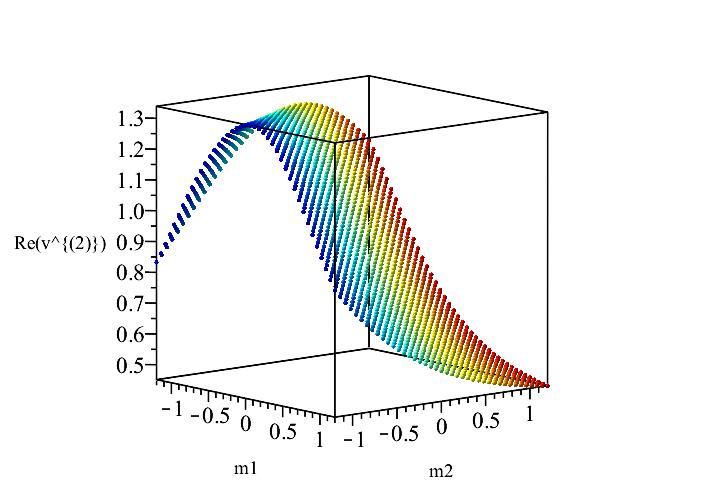}\\
  \caption{For genus $g=2$, local unreduced theta quotient for $u^{(2)}$ and $v^{(2)}$ in \eqref{eq:4.3}}\label{Fig-4}
\end{figure}

To show more details, we  present the quasi-periodic period-lattice structure in Figure \ref{Fig-2}. Over the enlarged lattice domain, repeated peaks, troughs, and band-like modulation patterns become clearly visible. These structures are spatially organized but are not exactly identical from one recurrence cell to another. Therefore, the pattern is more appropriately described as quasi-periodic rather than strictly periodic.
\begin{figure}[htbp]
  \centering
  \includegraphics[width=7cm]{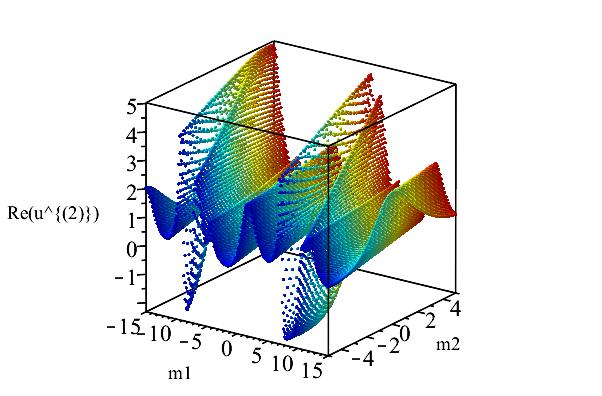}\quad
  \includegraphics[width=7cm]{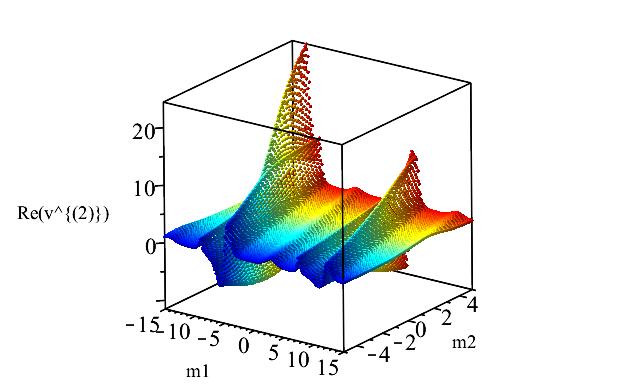}\\
  \caption{For genus $g=2$, theta period lattice for real part of $u^{(2)}$ and $v^{(2)}$ in \eqref{eq:4.3}}\label{Fig-2}
\end{figure}
\begin{figure}[htbp]
  \centering
  \includegraphics[width=7cm]{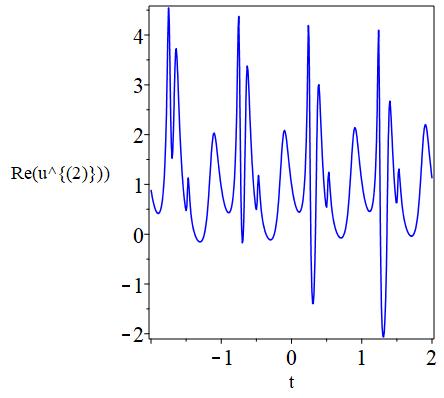}\quad
  \includegraphics[width=7cm]{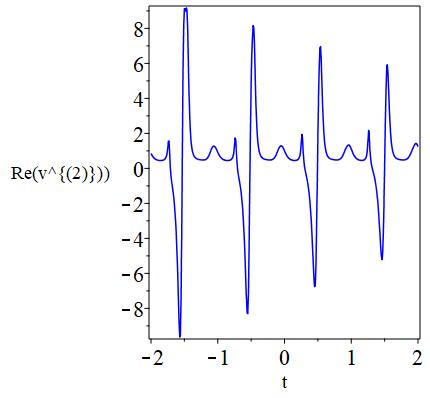}\\
  \caption{For genus $g=2$, slice along approximate recurrence (14,4) for real part of $u^{(2)}$ and $v^{(2)}$ in \eqref{eq:4.3}}\label{Fig-3}
\end{figure}

The Figure \ref{Fig-3} extracts a one-dimensional slice along an approximate recurrence direction of the period lattice in order to further verify the repeated structures observed in the two-dimensional plots. The principal peaks and troughs of the two theta-functional components recur approximately, but successive intervals do not coincide exactly. The nonzero recurrence error is consistent with the quasi-periodic nature of the finite-gap theta-functional solution and confirms that the repeated structures are not merely visualization artifacts.

The numerical results demonstrate pronounced directional dependence, two-dimensional coupled modulation, and quasi-periodic behavior in the genus-2 theta-functional representation. Although the two components share the same Riemann surface, period matrix, and basic Abelian phases, their distinct theta-functional structures lead to substantially different local geometries and amplitude responses. The period-lattice representations and the approximate recurrence slice further support that these patterns originate from the quasi-periodicity of the Riemann theta function rather than from a simple strict periodic repetition.


\section{Conclusions}\label{sec-5}
In this paper, we exhibited explicit quasi-periodic solutions for the discrete DNLS equation. The method we used is based on the nonlinearization technique and the theory of Riemann surface. We showed a Lax pair for the discrete DNLS equation, by which integrable symplectic maps are constructed. We found that the symplectic structure can be represented as the exterior differential of a differential one form. This investigation is especially helpful for the study of the $\mathrm{Q2, H2}$ and further models in the Adler-Bobenko-Suris hierarchy \cite{Adler}.
Furthermore, the quasi-periodic behaviors of the obtained solutions are visualized through graphical illustrations. This work extends the framework established in \cite{XCZ-2022,XZ} to the $g=2$ case.

At last, there are some related open problems. The discrete DNLS equation \eqref{eq:0.1} can't lead to the discretization of DNLS equation \eqref{KN} by
the reduction $v=-\delta \bar{u}$. The discretization procedures for integrable systems with complex conjugate are still waiting for establishing. In addition, in this work the infinity on the Riemann surface is a non-branching point. One may develop this approach to illustrate solutions defined on other types of Riemann surface.
Refer to \cite{Kundu,Kakei,Li}, the three kinds of DNLS equations can be linked through gauge transformations, and their unified solutions have been obtained. Whether one can unify algebro-geometric frameworks for discrete analogues of the DNLS equations will be an interesting problem.

\subsection*{Acknowledgments}
This work of X.X. Xu is supported by National Natural Science Foundation of China(NNSFC) (Grant Nos. 12571269, 11501521 ). The work of J.Y. Zhu is supported by NNSFC (Grant No. 12571268).

\subsection*{Declarations}

No data was used for the research described in the article.


\end{document}